\documentclass[10pt, conference, letterpaper]{IEEEtran}

\IEEEoverridecommandlockouts
\usepackage{amsfonts}
\usepackage{amssymb}
\usepackage{array, makecell}
\usepackage{mathtools}
\usepackage{graphicx}
\usepackage{subfigure}
\usepackage{enumerate}
\usepackage{amsmath}
\usepackage{color}
\usepackage{amsthm}
\usepackage{algorithm}
\usepackage{algpseudocode}
\usepackage{stfloats}
\usepackage{multirow}
\theoremstyle{definition}
\usepackage{color}
\usepackage{url}
\usepackage{dsfont}

\ifCLASSOPTIONcompsoc
  \usepackage[nocompress]{cite}
\else
  \usepackage{cite}
\fi

\newcommand{\bp}{\begin{proof} \small }
\newcommand{\ep}{\end{proof} \normalsize}
\newcommand{\epx}{\end{proof} \small}
\newcommand{\bpa}{\begin{proofappx} \footnotesize }
\newcommand{\epa}{\end{proofappx} \small }
\newtheorem{theorem}{Theorem}
\newtheorem{proposition}{Proposition}

\newtheorem{assumption}{Assumption}

\newtheorem*{theorem*}{Theorem}
\newtheorem*{proposition*}{Proposition}
\newtheorem*{corollary*}{Corollary}
\newtheorem*{lemma*}{Lemma}
\newtheorem*{assumption*}{Assumption}
\newtheorem*{definition*}{Definition}
\newtheorem*{claim*}{Claim}

\newcommand{\be}{\begin{equation}}
\newcommand{\ee}{\end{equation}}
\newcommand{\bs}{\begin{subequations}}
\newcommand{\es}{\end{subequations}}
\newcommand{\bq}{\begin{eqnarray}}
\newcommand{\eq}{\end{eqnarray}}
\newcommand{\bqn}{\begin{eqnarray*}}
\newcommand{\eqn}{\end{eqnarray*}}

\newcommand{\ba}{\left[ \begin{array}}
\newcommand{\ea}{\\ \end{array} \right]}
\newcommand{\ben}{\begin{enumerate}}
\newcommand{\een}{\end{enumerate}}

\def\real{{\mathchoice%
{\hbox{\rm\setbox1=\hbox{I}\copy1\kern-.45\wd1 R}}
{\hbox{\rm\setbox1=\hbox{I}\copy1\kern-.45\wd1 R}}
{\hbox{\scriptsize\rm\setbox1=\hbox{I}\copy1\kern-.45\wd1 R}}
{\hbox{\scriptsize\rm\setbox1=\hbox{I}\copy1\kern-.45\wd1 R}}}}

\def\Zint{{\mathchoice{\setbox1=\hbox{\sf Z}\copy1\kern-.75\wd1\box1}
{\setbox1=\hbox{\sf Z}\copy1\kern-.75\wd1\box1}
{\setbox1=\hbox{\scriptsize\sf Z}\copy1\kern-.75\wd1\box1}
{\setbox1=\hbox{\scriptsize\sf Z}\copy1\kern-.75\wd1\box1}}}
\newcommand{\complex}{ \hbox{\rm C\kern-0.45em\rule[.07em]{.02em}{.58em}%
\kern 0.43em}}

\makeatletter
\newcommand{\algmargin}{\the\ALG@thistlm}
\makeatother
\newlength{\whilewidth}
\algdef{SE}[parWHILE]{parWhile}{EndparWhile}[1]
{\parbox[t]{\dimexpr\linewidth-\algmargin}{%
		\hangindent\whilewidth\strut\algorithmicwhile\ #1\ \algorithmicdo\strut}}{\algorithmicend\ \algorithmicwhile}%
\algnewcommand{\parState}[1]{\State%
	\parbox[t]{\dimexpr\linewidth-\algmargin}{\strut #1\strut}}

\ifodd 1

\else

\fi

\def\BibTeX{{\rm B\kern-.05em{\sc i\kern-.025em b}\kern-.08em
		T\kern-.1667em\lower.7ex\hbox{E}\kern-.125emX}}

\usepackage[bottom=1.18in,left=0.65in,right=0.6in,top=0.67in]{geometry}

\begin{document}

\title{Adaptive User-Centric Entanglement Routing in Quantum Data Networks\\
}


\author{Lei Wang, Jieming Bian, 
    Jie Xu,~\IEEEmembership{Senior Member,~IEEE}
\thanks{Lei Wang, Jieming Bian and Jie Xu are with the Department of Electrical and Computer Engineering, University of Miami, Coral Gables,
FL 33146, USA. Email: \{lxw725, jxb1974, jiexu\}@miami.edu. 
}
}
\maketitle

\begin{abstract}
Distributed quantum computing (DQC) holds immense promise in harnessing the potential of quantum computing by interconnecting multiple small quantum computers (QCs) through a quantum data network (QDN). Establishing long-distance quantum entanglement between two QCs for quantum teleportation within the QDN is a critical aspect, and it involves entanglement routing — finding a route between QCs and efficiently allocating qubits along that route. Existing approaches have mainly focused on optimizing entanglement performance for current entanglement connection (EC) requests. However, they often overlook the user's perspective, wherein the user making EC requests operates under a budget constraint over an extended period. Furthermore, both QDN resources (quantum channels and qubits) and the EC requests, reflecting the DQC workload, vary over time. In this paper, we present a novel user-centric entanglement routing problem that spans an extended period to maximize the entanglement success rate while adhering to the user's budget constraint. To address this challenge, we leverage the Lyapunov drift-plus-penalty framework to decompose the long-term optimization problem into per-slot problems, allowing us to find solutions using only the current system information. Subsequently, we develop efficient algorithms based on continuous-relaxation and Gibbs-sampling techniques to solve the per-slot entanglement routing problem. Theoretical performance guarantees are provided for both the per-slot and long-term problems. Extensive simulations demonstrate that our algorithm significantly outperforms baseline approaches in terms of entanglement success rate and budget adherence.
\end{abstract}


\section{Introduction}
Quantum computing represents a groundbreaking paradigm with the potential to revolutionize computing and information processing\cite{steane1998quantum}. It offers unparalleled efficiency in solving specific problems that classical computers struggle with. However, in the foreseeable future, commodity Quantum Computers (QCs) will likely be constrained by a limited number of quantum bits, also known as qubits. Meanwhile, meaningful quantum computing applications may demand hundreds or even thousands of qubits\cite{cacciapuoti2019quantum}. To address this limitation, researchers have put forth the concept of Distributed Quantum Computing (DQC)\cite{cirac1999distributed, buhrman2003distributed, mao2023qubit}. This approach aims to distribute computational tasks among several smaller QCs, interconnected through a Quantum Data Network (QDN)\cite{yang2023asynchronous}. By leveraging this distributed setup, it becomes feasible to tackle more complex problems that demand a larger quantum resource while making the most of the available quantum computing capabilities. Different from quantum key distribution (QKD), which focuses on creating quantum bits (qubits) for cryptographic key delivery and has the capability to regenerate and retransmit qubits if needed \cite{mehic2020quantum}, QDNs encode data information in the data qubits and employ quantum teleportation \cite{bouwmeester1997experimental} techniques to address the retransmission challenge  due to the no-cloning theorem, which relies heavily on the stable long-distance entanglement \cite{wootters1982single}.

Quantum entanglement is a fundamental building block that plays a pivotal role in various applications\cite{vedral2014quantum}. To create quantum entanglement between two QCs, such as Alice and Bob, entangled Bell pairs of photons, known as qubits, are generated at one side and one of the entangled qubits is transmitted to the other side through a physical fiber-optic channel. However, the process faces challenges due to the inherent lossiness of the optical fiber, resulting in a success rate much less than one, which is influenced by the physical distance between Alice and Bob and the material properties of the fiber-optic channel \cite{stephenson2020high}. Establishing an entanglement link is probabilistic and unstable, relying on the availability of quantum channels between the parties and qubits on both ends. To enhance the probability of successful entanglement, more qubits can be allocated on both ends, and additional channels can be utilized if available, effectively increasing the number of attempts made simultaneously. Nevertheless, quantum channels are limited, and each quantum node has restricted quantum memory to store qubits, adding further complexity to the process. 

\begin{figure}
    \centering
    \includegraphics[width=0.35\textwidth]{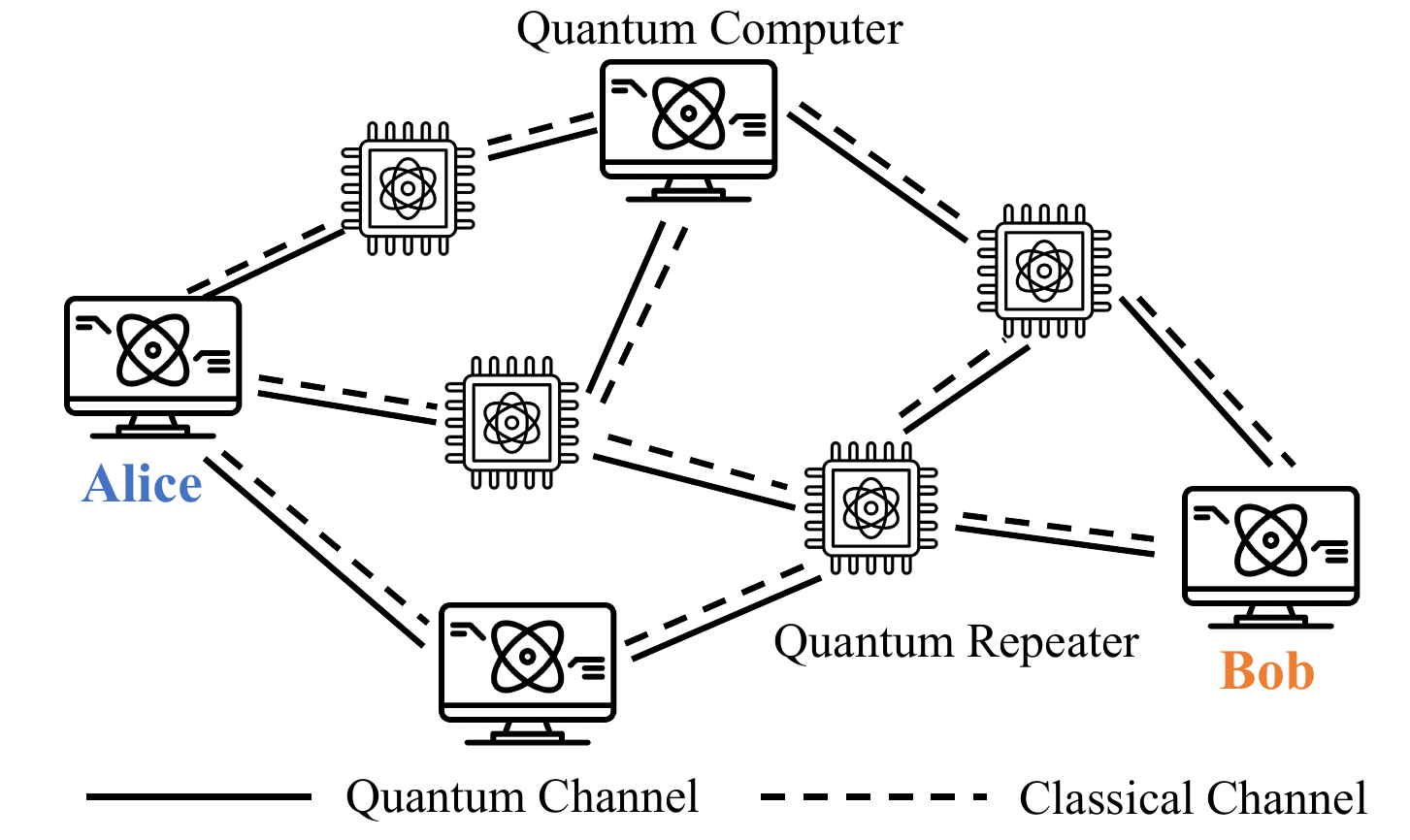}
     \vspace{-10pt}
    \caption{Quantum Data Network.}
    \label{fig:QDN}
    \vspace{-15pt}
\end{figure}

In QDNs, as illustrated in Figure~\ref{fig:QDN}, QCs are often geographically distant from each other and lack direct channel connections. Instead, they are interconnected through several other QCs or quantum repeaters (QRs), enabling the establishment of entanglement links only between adjacent quantum nodes. Fortunately, a process known as quantum swapping can be employed at intermediary quantum nodes when both Alice and Bob are connected to the same node\cite{jennewein2001experimental}. This process allows the creation of entanglement between Alice and Bob. By repeating this operation for multiple quantum links along a route consisting of nodes between Alice and Bob, a long-distance entanglement connection (EC) can be established. Though this operation may fail, we can still assume its probability approaches 1 due to the most advanced research and related work \cite{yang2023asynchronous, hilaire2021error}. This approach enables effective entanglement distribution over considerable distances in QDN. Optimizing the performance of QDN requires meticulous entanglement routing, consisting both route selection and qubit allocation, especially in light of the constraints posed by limited quantum channels and qubit availability. By strategically managing these valuable resources, QDNs can achieve enhanced efficiency, robustness, and overall effectiveness in their operations.

\subsection{Related Work}
Several researches have been dedicated to addressing the entanglement routing problem, which involves establishing long-distance ECs for various source-destination (SD) pairs, while considering route selection and resource allocation. Early works focused on specialized network topologies, including sphere \cite{schoute2016shortcuts}, grid \cite{pant2019routing}, ring \cite{chakraborty2019distributed}, and star \cite{vardoyan2019stochastic} configurations. These studies explored specific network structures to understand and optimize entanglement routing. In more recent studies, researchers have considered a general QDN setting, where ECs are requested by multiple SD pairs \cite{shi2020concurrent}. The primary objective in these studies is to maximize the expected network throughput, although performance guarantees are not guaranteed. A further extension of the research \cite{zeng2022multi} has sought to maximize both throughput and the number of SD pairs served by the QDN, accompanied by theoretical analyses. This work addresses the challenge of efficiently utilizing network resources to support multiple SD pairs concurrently. To enhance failure tolerance, some approaches \cite{zhao2021redundant, zhao2022segmented, zhao2022e2e} have leveraged redundant entanglement links in routing. This redundancy ensures that even if certain links fail, alternative routes can be utilized, increasing the robustness of the QDN. Several studies incorporate the consideration of entanglement connection, a crucial metric for assessing the quality of remote entanglement links, into the framework of entanglement routing \cite{caleffi2017optimal, zhao2022e2e, li2022fidelity}. A qubit allocation algorithm in QDNs was proposed by combining simulated-annealing and local search \cite{mao2023qubit}. The entanglement routing scenario in QDNs was extended beyond the time slot mode to an asynchronous scheme in \cite{yang2023asynchronous} and an online mode that processes requests upon arrival, termed as online entanglement routing \cite{yang2022online}, both resulting in proactively utilize idle quantum resources more effectively. Additionally, opportunism has been introduced to QDNs \cite{farahbakhsh2022opportunistic}, allowing for the opportunistic establishment of quantum links and enhancing routing flexibility.

\subsection{Our Contribution}
However, the majority of existing works on entanglement routing primarily focus on myopically optimizing entanglement performance in the current time slot (more precisely, the current EC requests). A critical aspect is often overlooked: the fact that ECs are user-requested, and the entire QDN is established and maintained by the service provider. Consequently, allocating qubits to generate entanglement links for EC requests during a time slot incurs a cost charged by the QDN provider. The allocated qubits are valuable resources that cannot be utilized by other users during that specific time slot, making efficient resource allocation essential. Because users typically operate within a budget to avoid excessive expenses over an extended period, spanning multiple time slots, strategically spending this budget in the long-run is crucial for the overall entanglement performance. On the other hand, the entanglement routing problem is significantly complicated by the fact that EC requests may follow a random process unknown to the user beforehand, which may depend on the DQC workload arrival process. As a result, entanglement routing decisions are inherently coupled across many time slots. Therefore, it becomes imperative to develop entanglement routing strategies that take into account the users' specific requests and budget limitations in the long term while also ensuring the efficient utilization of QDN resources.

In this paper, we propose a novel user-centric entanglement routing approach that balances between maximizing the entanglement connection success rates and optimizing the overall user cost over time. Our main contributions are as follows:
\begin{itemize}
    \item We formulate a user-centric entanglement routing problem that jointly considers route selection and qubit allocation over an extended period to maximize the entanglement success rate, while respecting the user's budget constraint.
    \item To solve this problem efficiently, we leverage the Lyapunov drift-plus-penalty framework, decomposing the long-term optimization into per-slot problems. These per-slot problems can be solved using only current system information, without requiring knowledge of future EC request, quantum channel, and qubit availability. We further design an efficient algorithm to address the challenging integer program of per-slot entanglement routing.
    \item We provide theoretical performance guarantees for our algorithms, for both the per-slot and the long-term problems.
    \item Through extensive simulations, we demonstrate the efficiency and superiority of our approach, showcasing its significant outperformance compared to baseline methods.
\end{itemize}

\section{Background} \label{s2}
Before delving into our system model and presenting the entanglement routing problem, we provide some essential background information on QDN as a foundation.

\subsubsection{Qubit}
A qubit, or quantum bit, serves as the fundamental building block of quantum information. Unlike classical binary bits confined to 0 or 1 states, a qubit can exist in a coherent superposition of these states, represented as $\alpha|0\rangle + \beta|1\rangle$ \cite{van2014quantum}. There are various physical implementations of qubits, but for quantum communication transmitted through optical fiber-based quantum channels, photonic qubits emerge as the most promising candidates \cite{mattle1996dense}.

\subsubsection{Quantum Entanglement}
Quantum entanglement is a remarkable state that involves multiple qubits and cannot be simply represented by the product of their individual states. The measurement outcomes of entangled qubits are correlated, leading to intriguing quantum phenomena. In a two-qubit system, a pair of entangled qubits, denoted as A and B, form what are known as Bell pairs. These Bell pairs can exist in one of two standard bases: $\frac{|0_A 0_B\rangle \pm |1_A 1_B\rangle}{\sqrt{2}}$ or $\frac{|0_A 1_B\rangle \pm |1_A 0_B\rangle}{\sqrt{2}}$ \cite{nielsen2001quantum}.

\begin{figure}
    \centering
    \includegraphics[width=0.35\textwidth]{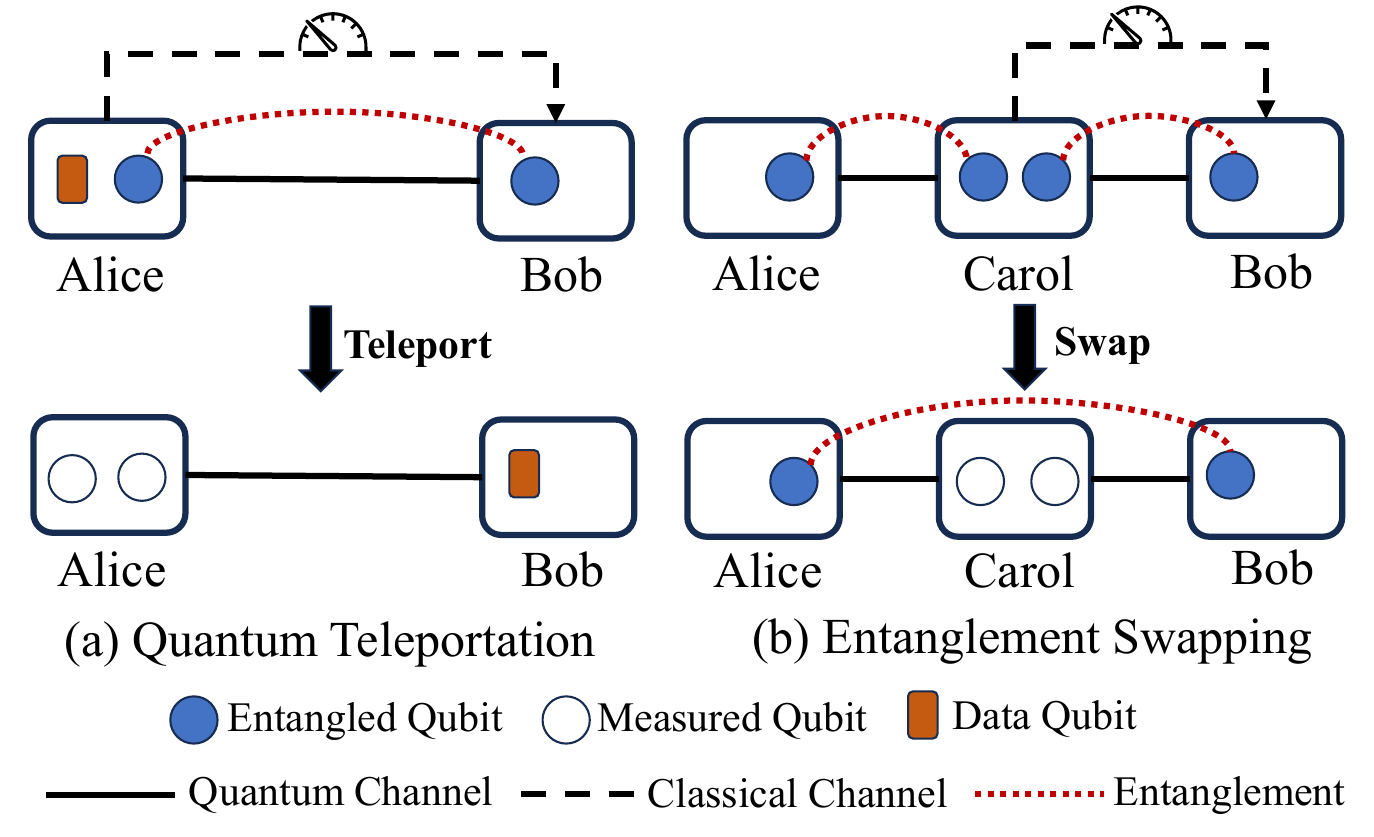}
    \vspace{-13pt}
    \caption{Quantum Teleportation and Entanglement Swapping.}
    \label{fig:quantum_opreation}
    \vspace{-10pt}
\end{figure}
\subsubsection{Quantum Teleportation}
Quantum teleportation is a significant application of quantum entanglement and the enabler of DQC. As illustrated in Figure~\ref{fig:quantum_opreation}, when Alice and Bob share an entangled pair of qubits (referred to as ebit), Alice can teleport the state of another qubit with data information (known as dbit) to Bob \cite{bouwmeester1997experimental}. This process involves Alice performing a joint measurement (Bell State Measurement) of her dbit and the shared ebit and then communicating the measurement result to Bob through a classical channel. Upon receiving the measurement result, Bob applies certain unitary operations on his own ebit \cite{pirandola2015advances}. As a result, Bob's ebit obtains the state of the original dbit, while Alice's dbit collapses, and the entanglement between the two ebits is destroyed. This phenomenon allows for the transfer of quantum information without the physical movement of particles. 

\subsubsection{Entanglement Swapping}
Entanglement swapping is a crucial technique for creating long-distance entanglement. The process of entanglement swapping is depicted in Figure~\ref{fig:quantum_opreation}. Here is an example scenario: Alice shares an entangled pair of qubits with a third party, Carol, while Bob also shares a qubit pair with Carol. By performing the swapping operation on her qubits, Carol can teleport the state of her qubit, initially entangled with Alice, to Bob. Consequently, Alice and Bob effectively share an entangled pair of qubits, even though they are not directly connected to each other. This enables the establishment of long-distance entanglement between distant parties \cite{jennewein2001experimental}. For the purpose of this work, we assume a successful swapping operation, as recent advancements have significantly increased its success rate to approximately one \cite{hilaire2021error}, which is also a reasonable setting in the state-of-the-art work \cite{yang2023asynchronous}. Moreover, the failure probability of swapping can also be considered as part of the overall failure probability of establishing entanglement connections, just incorporating a product term in Equation~\ref{eq:rate}.

\subsubsection{Quantum Data Network}
In the QDN, quantum nodes represent nodes that can be either QCs or QRs. These quantum nodes are interconnected through quantum channels, forming edges in the QDN graph. All quantum nodes possess the capability to perform entanglement swapping and establish entanglement with other nodes, but have a limited qubit capacity due to quantum memory constraints. Additionally, quantum channels are inherently lossy, and the success rate of a single attempt to create entanglement can be as low as $2.18 \times 10^{-4}$ \cite{stephenson2020high}.
To establish long-distance entanglement, a route must be determined from the source to the destination. Subsequently, a sequence of quantum links along this route is generated, and swapping operations are performed at nodes along the route. The typical time for entanglement to decohere is approximately $1.46$ seconds \cite{dahlberg2019link}, while the time required for a single entanglement attempt is around $165\mu$s \cite{dahlberg2019link}. In a time slot, defined as the entanglement duration, thousands of attempts can be made for a single quantum link. Efficient utilization of qubits and quantum channels is critical for the success of long-distance entanglement establishment in the QDN.

\section{System Model} \label{s3}
\subsection{Quantum Data Network}
We consider a QDN represented by an undirected graph $\mathcal{G} = \langle\mathcal{V}, \mathcal{E}\rangle$, where $\mathcal{V}$ is the set of quantum nodes and $\mathcal{E}$ is the set of edges. Each edge $e = (u,v) \in \mathcal{E}$ connects two nodes $v$ and $u$. Each quantum node $v \in \mathcal{V}$ is equipped with a limited number of $Q_v$ qubits. However, the available qubits $Q^t_v$ can change over time, denoted by $t$, as some qubits may be occupied by other users. This occupancy is considered as an exogenous process. In order for two quantum nodes $v, u \in \mathcal{V}$ to be connected by an edge $e \in \mathcal{E}$, there must be at least one quantum channel (i.e., physical fiber-optical wire) between them. Let $W^t_e$ represent the number of available quantum channels on edge $e$, which can also vary over time depending on the usage by other users.

\subsection{Quantum Entanglement Link}
In the QDN, a quantum link can be established on an edge $e = (v, u)$ using one qubit on node $v$, one qubit on node $u$, and a quantum channel on $e$. However, successful entanglement establishment is not guaranteed on every quantum channel. Let $\tilde{p}_e$ denote the success probability of establishing entanglement on one quantum channel between nodes $v$ and $u$ during a single attempt. This probability depends on both the physical properties of the channel material and the length of the quantum channel. Typically, $\tilde{p}_e$ is low. To increase the probability of successful entanglement, nodes $v$ and $u$ can utilize multiple quantum channels and make multiple attempts on each channel within a given time slot. Assuming that the outcomes of these attempts are independent, the success probability on a single channel after $A$ attempts is given by $p_e = 1 - (1-\tilde{p}_e)^A$.

Now, if $n_e$ channels are used for establishing a quantum link between $v$ and $u$, the overall success probability is given by:
\begin{align}
P_e(n_e) = 1 - (1 - p_e)^{n_e},
\end{align}
where $P_e(n_e)$ represents the probability of successfully establishing a quantum link using $n_e$ channels.

\subsection{Problem Formulation}
We are primarily focused on addressing the entanglement routing for a single quantum user within a QDN. This user's objective is to establish ECs for between source-destination (SD) pairs over a specified duration of time slots, denoted as $T$. In each time slot $t$, the user aims to find routes among the QDN for a particular set of SD pairs, denoted as $\Phi^t$, which may depend on the DQC requirements. We assume an upper bound $F$ on the size of $\Phi^t, \forall t$. It's important to note that the set of SD pairs may vary over time.

In our scenario, each SD pair, denoted as $\phi \in \Phi^t$, consists of a source node, represented by $s(\phi)$, and a destination node, represented by $d(\phi)$. Additionally, there exists a set of potential routes, denoted as $\mathcal{R}(\phi)$, associated with each SD pair $\phi$. We assume an upper bound $R$ on the size of this set. 
It's important to highlight that the candidate set $\mathcal{R}(\Phi)$ can be pre-computed by choosing routes with shorter lengths/hops to minimize its size. Alternatively, any established shortest path finding algorithm, such as Dijkstra's Algorithm, can be employed for this purpose \cite{dijkstra2022note}. A route $r \in \mathcal{R}(\phi)$ is defined as a subset of graph edges, denoted as $\mathcal{E}$, that form a connected route between the source node $s(\phi)$ and the destination node $d(\phi)$. We assume an upper bound $L$ on the length of a route $r$. To establish EC between the source and destination, qubits need to be allocated among the nodes along the selected route.

Given a specific qubit allocation $\mathcal{N}(r) = \{n_e(r), \forall e \in r\}$ for a route $r$, the entanglement success rate can be calculated as the product of the success probabilities of the individual edges on that route as follows:
\begin{align}
\label{eq:rate}
P(r, \mathcal{N}(r)) = \prod_{e \in r} P_e(n_e(r)).
\end{align}
Here, $P_e(n_e(r))$ denotes the success probability of edge $e$ when $n_e(r)$ qubits are allocated on $e$. It is worth noting that although we assume each SD pair makes a single EC request between the source and destination, the extension to multiple EC requests from a single SD pair is straightforward. In such cases, we can treat each entanglement connection request as a separate SD pair, each with a single EC request.
 
It is important to note that the success rate of entanglement swapping can also be viewed as a product term in Equation~\ref{eq:rate} for establishing entanglement connections. For simplicity, we choose to disregard the impact of entanglement swapping, given recent advancements that have substantially elevated its success rate to approximately one \cite{hilaire2021error}. This assumption aligns with the current state-of-the-art work as well \cite{yang2023asynchronous}.

\textbf{Objective}: The quantum user's objective is to optimize the entanglement success rate over $T$ time slots ($t = 0, 1, ..., T-1$) by selecting a route for each SD pair and allocating qubits along those routes. In order to ensure fairness among the SD pairs and distribute quantum network resources appropriately, we adopt the concept of proportional fairness \cite{kelly1997charging}. The goal is to maximize the following objective function:
\begin{align}
\sum_{t=0}^{T-1}\sum_{\phi\in \Phi^t}\log P(r^t(\phi), \mathcal{N}^t(r^t(\phi))).
\end{align}
Here, $r^t(\phi)$ represents the chosen route for SD pair $\phi$ in time slot $t$, and $\mathcal{N}^t(r^t(\phi))$ denotes the qubit allocation along this route. The objective function sums over all time slots and SD pairs, and the logarithm of the entanglement success rate is used to capture proportional fairness.

\textbf{Capacity Constraints}: In order to solve the aforementioned maximization problem, it is necessary to take into account the capacity constraints of the quantum network. These constraints arise from varying qubit capacity $Q^t_v$ of each node $v \in \mathcal{V}$ and varying quantum channel capacity $W^t_e$ of each edge $e \in \mathcal{E}$.

The qubit capacity constraint can be expressed as follows:
\begin{align}
\sum_{\phi \in \Phi^t} \sum_{e \in r^t(\phi)} \mathds{1}_e(v) n^t_e(r^t(\phi)) \leq Q^t_v, \quad \forall v \in \mathcal{V}, \quad \forall t.
\label{qubit_cap}
\end{align}
This equation ensures that the total number of qubits allocated to node $v$ from the selected routes $r^t(\phi)$ of all SD pairs $\phi$ in time slot $t$ does not exceed the qubit capacity $Q^t_v$ of that node. Similarly, the quantum channel capacity constraint can be formulated as:
\begin{align}
\sum_{\phi \in \Phi^t} \mathds{1}_{r^t(\phi)}(e) n^t_e(r^t(\phi)) \leq W^t_e, \quad \forall e \in \mathcal{E}, \quad \forall t.
\label{channel_cap}
\end{align}
This inequality ensures that the total number of qubits allocated to edge $e$ from the selected routes $r^t(\phi)$ of all SD pairs $\phi$ in time slot $t$ does not exceed the quantum channel capacity $W^t_e$ of that edge. In both equations, the indicator function $\mathds{1}_\mathcal{X}(x)$ is utilized, where $\mathds{1}_\mathcal{X}(x) = 1$ if $x \in \mathcal{X}$ and $\mathds{1}_\mathcal{X}(x) = 0$ otherwise.

\textbf{Budget Constraint}: In addition to capacity constraints, the usage of the quantum network incurs costs for the quantum user. These costs are associated with the utilization of qubits and the establishment of quantum links along the routes. Assuming that the cost is proportional to the number of qubits and quantum channels used, the cost in time slot $t$ can be expressed as $\sum_{\phi \in \Phi^t}\sum_{e \in r^t(\phi)} n^t_e(r^t(\phi))$.

To incorporate the cost aspect, we consider a total budget $C$ that represents the user's allowance for using the quantum network over a period of $T$ time slots. The budget constraint is imposed as follows:
\begin{align}
\sum_{t=0}^{T-1} \sum_{\phi \in \Phi^t}\sum_{e \in r^t(\phi)} n^t_e(r^t(\phi)) \leq C.
\label{budget_cap}
\end{align}
This constraint ensures that the total cost, which is the sum of the number of qubits and quantum channels used across all time slots and SD pairs, does not exceed the specified budget.

\textbf{Fidelity Constraint}: Our work focuses primarily on user-centric entanglement routing, taking into account a long-term limited budget thus fidelity is a secondary additional requirement in our scenario. Previous research typically utilizes entanglement fidelity to assess the quality of EC \cite{zhao2022e2e, li2022fidelity}. In fact, we can easily integrate a constraint into \textbf{P1}, which calculates the fidelity of the chosen route and ensures it remains below the fidelity target in each time slot. This constraint is analogous to aforementioned capacity constraints. Since it is a per-slot constraint independent of long-term knowledge, which is different from the budget constraint, we can easily modify the per-slot problem \textbf{P2} by incorporating this constraint while still utilizing our proposed algorithm outlined in Section~\ref{s4}. Without affecting our key idea, we do not consider fidelity for simplicity, which is a reasonable formulation in many recent studies \cite{yang2023asynchronous,zhao2022segmented,yang2022online,zeng2022multi}. 


In summary, the quantum routing problem can be formulated as follows:
\label{p1}
\begin{align}
\textbf{P1:}\quad \max & \quad\sum_{t=0}^{T-1}\sum_{\phi\in \Phi^t}\log P(r^t(\phi), \mathcal{N}^t(r^t(\phi))) \nonumber\\
\text{s.t.} & \quad \text{Capacity constraints:} \quad \eqref{qubit_cap}, \eqref{channel_cap} \nonumber\\
& \quad \text{Budget constraint:}  \quad \eqref{budget_cap} \nonumber\\
&\quad r^t(\phi) \in \mathcal{R}(\phi), \forall \phi \in \Phi^t, \forall t \nonumber\\
&\quad n_e(r^t(\phi)) \in \mathbb{Z}_{++}, \forall e \in r^t(\phi), \forall \phi \in \Phi^t, \forall t \nonumber,
\end{align}
where $\mathbb{Z}_{++} \triangleq \{1, 2, ...\}$ are positive integers  to ensure connectivity. It is important to note that the capacity constraints $\eqref{qubit_cap}$ and $\eqref{channel_cap}$ are short-term constraints that must be satisfied in each time slot $t$, while the budget constraint $\eqref{budget_cap}$ is a long-term constraint that considers the cumulative usage over $T$ time slots. We also denote $\sum_{\phi \in \Phi^t} \log P(r^t(\phi), \mathcal{N}^t(r^t(\phi))) = u(\boldsymbol{r}^t, \mathcal{N}^t)$ for simplicity. 

\textbf{Challenges}: There are two primary challenges associated with directly solving the problem \textbf{P1}. \textit{Firstly}, the decisions regarding quantum route selection and qubit allocation have a correlated impact across different time slots, affecting both the objective function and the constraints. Allocating a larger qubit budget in the current time slot can enhance the success rate of entanglement, but it may potentially degrade performance in future time slots. However, since we lack prior knowledge of future EC requests, qubit availability, and channel capacity, we require an online algorithm that can make decisions without relying on such information. \textit{Secondly}, within each time slot $t$, there is a sequential order that determines the chosen quantum route and subsequently the allocation of qubits along that route. The problem itself is complex due to the discrete nature of the decision variables and the vast decision space involved. Hence, it is necessary to employ a low-complexity algorithm capable of efficiently performing route selection and qubit allocation within a time slot.

\section{Online User-Centric Entanglement Routing}\label{s4}

To address the aforementioned challenges, we propose an online algorithm, called \underline{O}nline u\underline{S}er-\underline{C}entric ent\underline{A}nglement \underline{R}outing (OSCAR), that breaks down the long-term problem of route selection and qubit allocation into per-slot problems. Subsequently, we have developed an efficient algorithm to solve each per-slot problem individually.

\subsection{Long-Term Problem Decomposition}
Our approach is based on the Lyapunov drift-plus-penalty framework, which utilizes a virtual cost deficit queue $q^t$ to guide the decisions of route selection and qubit allocation in each time slot, ensuring adherence to the long-term budget constraint. For simplicity, we denote $c^t = \sum_{\phi \in \Phi^t}\sum_{e \in r^t(\phi)} n^t_e(r^t(\phi))$ as the cost incurred in time slot $t$. The virtual queue $q^t$ evolves according to the following recursion:
\begin{align}
\label{queue}
q^{t+1} = \max\{0, q^t + c^t - C/T\}.
\end{align}
Intuitively, the virtual queue captures the accumulated violation of the budget constraint. Thus, our objective is to maximize the entanglement success rate while minimizing the length of the virtual queue. We define a constant $V > 0$ that will be discussed further in Section \ref{sec:ablation}.

In each time slot, we formulate the following drift-plus-penalty maximization problem, denoted as \textbf{P2}:
\begin{align}
\textbf{P2:} \quad \max & \quad V\cdot\sum_{\phi\in \Phi^t}\log P(r^t(\phi), \mathcal{N}^t(r^t(\phi))) \nonumber\\
&- q^t\cdot \sum_{\phi \in \Phi^t}\sum_{e \in r^t(\phi)} n^t_e(r^t(\phi)) \nonumber\\
\text{s.t.} & \quad \text{Capacity constraints:} \quad \eqref{qubit_cap}, \eqref{channel_cap} \nonumber\\
&\quad r^t(\phi) \in \mathcal{R}(\phi), \forall \phi \in \Phi^t, \forall t \nonumber\\
&\quad n_e(r^t(\phi)) \in \mathbb{Z}_{++}, \forall e \in r^t(\phi), \forall \phi \in \Phi^t, \forall t \nonumber.
\end{align}
For simplicity, we define $f(\boldsymbol{r}^t(\Phi^t), \mathcal{N}^t(\boldsymbol{r}^t(\Phi^t)))$ as the objective function. It is important to note that while solving \textbf{P2}, we consider the virtual queue length $q^t$, as well as the qubit capacity $Q^t_v$ for all $v\in\mathcal{V}$ and the channel capacity $W^t_e$ for all $e \in \mathcal{E}$, as given variables. Thus, the problem \textbf{P2} is a per-slot problem that solely relies on the available current information, without requiring future statistics of EC requests and qubit/channel capacity. By augmenting the original objective function (entanglement success rate) with the cost term weighted by the virtual queue, we dynamically balance the maximization of performance and the minimization of cost. The algorithm is summarized in Algorithm~\ref{alg:long_term} and its performance will be further analyzed in Section \ref{s5}.

\begin{algorithm}[t]
	\caption{OSCAR}
	\begin{algorithmic}[1]
		\State \textbf{Input}: $q^0$ and $V$
        \State \textbf{Output}: $r^t(\phi), \mathcal{N}^t(r^t(\phi)), \quad  \forall t \in [0, \dots, T]$
        \For {$t = 1, 2, ..., T$}
            \State Observe $\Phi^t$, $Q^t_v, \forall v \in \mathcal{V}$ and $W^t_e, \forall e \in \mathcal{E}$
            \State Solve \textbf{P2} (See Section \ref{sec:per-slot})
            \State Update virtual queue $q^t$ according to \eqref{queue}
        \EndFor
	\end{algorithmic}\label{alg:long_term}
\end{algorithm}

\subsection{Solving the Per-Slot Problem}
\label{sec:per-slot}
While \textbf{P2} does not rely on future information, it remains a challenging problem due to its large decision space as an integer program. To address this challenge, we initially focus on solving the qubit allocation problem with a fixed route selection for the SD pairs in $\Phi^t$. By obtaining the ``optimal'' qubit allocation for each route selection, we can subsequently determine the ``optimal'' quantum routes. For the sake of simplicity, we will omit the time index $t$ in this subsection.
\begin{algorithm}[t]
	\caption{Qubit Allocation}
	\begin{algorithmic}[1]
		\State \textbf{Input}: Route selection $\boldsymbol{r}(\Phi) = \{r(\phi): \phi \in \Phi\}$
        \State \textbf{Output}: Rounded solution  $\mathcal{N}^*$ 
        \State Obtain the solution $\tilde{\mathcal{N}}^*$ by addressing \textbf{P2}, wherein the integer restrictions $n_e \in \mathbb{Z}_{++}$ are relaxed to $n_e \geq 1$.
        \State ``Down-round'' and allocate surplus
        
	\end{algorithmic}\label{alg:qubit_allocation}
\end{algorithm}

\subsubsection{\textbf{Qubit Allocation}}
Our approach to tackle the qubit allocation problem involves utilizing continuous relaxation. Specifically in Algorithm~\ref{alg:qubit_allocation}, for a given route selection $\boldsymbol{r}(\Phi) = \{r(\phi): \phi \in \Phi\}$, we solve \textbf{P2} (with a fixed $\boldsymbol{r}(\Phi)$) by relaxing the integer constraints $n_e \in \mathbb{Z}_{++}$ to $n_e \geq 1$. Let $\tilde{\mathcal{N}}^* = \{\tilde{n}^*_e \in \mathbb{R}: \forall e \in r(\phi), \forall \phi\in \Phi\}$ be the optimal solution obtained from solving the relaxed problem. We then apply rounding to each $\tilde{n}^*_e$ to obtain an integer solution, denoted as $n^*_e$. To ensure the satisfaction of capacity constraints, we initially perform a ``down-rounding'' operation on $\tilde{n}^*_e$ for every edge $e$, and then allocate any surplus to the edges on the node or edge if possible. This rounding strategy guarantees that the rounded solution $\mathcal{N}^*$ adheres to the capacity constraints, and for all $e \in r(\phi)$, $\phi\in \Phi$, we have:
\begin{align}
n^*_e \geq 1 \quad \text{and} \quad \tilde{n}^*_e - n^*_e \leq 1.
\label{round_diff}
\end{align}

Next, we analyze the sub-optimality gap using the continuous relaxation approach. We begin by demonstrating that the continuous-relaxed problem \textbf{P2} (with a fixed $\boldsymbol{r}(\Phi)$) is a convex optimization problem.

\begin{proposition}
The continuous-relaxed problem \textbf{P2} with a fixed $\boldsymbol{r}(\Phi)$ is a convex optimization problem.
\end{proposition}
\begin{proof}
Since the constraints in the problem are linear, it suffices to demonstrate that the objective function is concave. For a given route selection $\boldsymbol{r}(\Phi) = \{r(\phi): \phi \in \Phi\}$, the objective function in \textbf{P2} can be expressed as:
\begin{align}
&V\cdot \sum_{\phi \in \Phi}\sum_{e \in r(\phi)} \log P_e(n_e(r(\phi)) - q \cdot \sum_{\phi \in \Phi} \sum_{e \in r(\phi)} n_e(r(\phi))\nonumber \nonumber\\
=&\sum_{\phi\in \Phi}\sum_{e\in r(\phi)}\left( V\cdot \log P_e(n_e(r(\phi)) - q\cdot n_e(r(\phi))\right) \nonumber.
\end{align}
Thus, it suffices to prove that $\log P_e(n_e)$ is concave. Since $1-p_e \in (0, 1)$, $P_e(n_e)$ is concave. According to the composition rule \cite{boyd2004convex}, $\log P_e(n_e)$ is also concave. This concludes the proof.
\end{proof}

Let $f(\mathcal{N})$ denote the objective function of \textbf{P2} when the route selection is fixed, and $\mathcal{N}^\text{opt}$ be the optimal integer solution. 

\begin{proposition}
\label{prop1}
The sub-optimality gap between $\mathcal{N}^\text{opt}$ and $\mathcal{N}^*$ is upper-bounded by $f(\mathcal{N}^\text{opt}) - f(\mathcal{N}^*) \leq VFL\log(2 - p^\text{min}) \triangleq \Delta$, where $p^\text{min} = \min_{e\in\mathcal{E}} p^e$. 
\end{proposition}
\begin{proof}
We examine the two terms in the objective function separately and denote:
\begin{align}
f_1(\mathcal{N}) &= V\cdot \sum_{\phi \in \Phi}\sum_{e \in r(\phi)} \log P_e(n_e(r(\phi))), \\
f_2(\mathcal{N}) &= q\cdot \sum_{\phi \in \Phi}\sum_{e \in r(\phi)} n_e(r(\phi)).
\end{align}
Consider the optimal solution $\tilde{\mathcal{N}}^*$ of the continuous-relaxed problem and its rounded solution $\mathcal{N}^*$. 
\begin{align}
&f_1(\tilde{\mathcal{N}}^*) - f_1(\mathcal{N}^*) \nonumber\\
=& V\cdot \sum_{\phi \in \Phi}\sum_{e \in r(\phi)} \left(\log P_e(\tilde{n}^*_e(r(\phi))) - \log P_e(n^*_e(r(\phi)))\right) \nonumber\\
\leq & V\cdot \sum_{\phi \in \Phi}\sum_{e \in r(\phi)}\left(\log P_e(2) - \log P_e(1)\right) \label{ineq_1}\\
=& V\cdot \sum_{\phi \in \Phi}\sum_{e \in r(\phi)} \log (2 - p_e) \nonumber\\
\leq & VFL\log(2 - p^\text{min}).
\label{ineq_2}
\end{align} 
Inequality \eqref{ineq_1} is due to the concavity of the logarithm, and it follows from the relation \eqref{round_diff}. Inequality \eqref{ineq_2} uses the monotonicity of the logarithm. 

On the other hand, it is obvious that $f_2(\tilde{\mathcal{N}}^*) - f_2(\mathcal{N}^*) \geq 0$ due to the down-rounding operation. Therefore, we have $f(\tilde{\mathcal{N}}^*) - f(\mathcal{N}^*) \leq VFL\log(2 - p^\text{min})$. Moreover, since $\tilde{\mathcal{N}}^*$ solves the relaxed problem, we have $f(\tilde{\mathcal{N}}^*) \geq f(\mathcal{N}^\text{opt})$ and hence $f(\mathcal{N}^\text{opt}) - f(\mathcal{N}^*) \leq VFL\log(2 - p^\text{min})$. 
\end{proof}
Proposition~\ref{prop1} demonstrates that our continuous-relaxation approach produces a $\Delta$-optimal solution for the per-slot problem \textbf{P2} with any given selected routes $\boldsymbol{r}(\Phi)$. The remaining step to complete the solution for \textbf{P2} is to determine the route for each $\phi \in \Phi$.

\begin{algorithm}[t]
	\caption{Route Selection}
	\begin{algorithmic}[1]
    \State \textbf{Input}: Initial route selection $\boldsymbol{r}^0$
    \State \textbf{Output}: Optimal route selection $\boldsymbol{r}^*$ 
		 \For {$k = 1, 2, ...$} until stable
            \State Randomly select a SD pair $\phi$
            \State Virtually modify its selection $\tilde{r}(\phi)$
            \State $\tilde{\boldsymbol{r}}^k \leftarrow (\tilde{r}(\phi), \{r^{k-1}(\phi')\}_{\phi' \neq \phi})$
            \State Allocate the qubits for $\tilde{\boldsymbol{r}}^k$ by Algorithm \ref{alg:qubit_allocation}
            \State Compute probability $\eta$ by \eqref{eta}
            \State With probability $1-\eta$, keep $\boldsymbol{r}^k = \boldsymbol{r}^{k-1}$; With probability $\eta$, change $\boldsymbol{r}^k = \tilde{\boldsymbol{r}}^k$
        \EndFor
	\end{algorithmic}\label{alg:route_selection}
\end{algorithm}

\subsubsection{\textbf{Route Selection}}
With the computed qubit allocation $\mathcal{N}^*(\boldsymbol{r})$ for any given route selection $\boldsymbol{r}$, determining the route is a straightforward process:
\begin{align}
\boldsymbol{r}^* = \arg\max_{\boldsymbol{r} \in \bigtimes_{\phi\in \Phi} \mathcal{R}(\phi)} f(\boldsymbol{r}, \mathcal{N}^*(\boldsymbol{r})).
\end{align}
In other words, we perform an exhaustive search on all possible route combinations for the SD pairs in $\Phi$ and select the combination with the highest per-slot objective value by applying the qubit allocation algorithm. Note that the candidate set $\mathcal{R}(\Phi)$ can be computed beforehand by selecting routes with shorter lengths/hops to reduce its size or using any other existing shortest path finding algorithm such as Dijkstra's Algorithm \cite{dijkstra2022note}. Let $\boldsymbol{r}^\text{opt}$ and $\mathcal{N}^\text{opt}(\boldsymbol{r}^\text{opt})$ be the optimal joint route selection and qubit allocation solution. Then, the solution produced by our approach, namely $\boldsymbol{r}^*$ and $\mathcal{N}^*(\boldsymbol{r}^*)$, is also $\Delta$-optimal, i.e.,
\begin{align}
f(\boldsymbol{r}^\text{opt}, \mathcal{N}^\text{opt}(\boldsymbol{r}^\text{opt})) - f(\boldsymbol{r}^*, \mathcal{N}^*(\boldsymbol{r}^*)) \leq \Delta.
\end{align}

However, due to the combinatorial nature of the route space for each SD pair in $\Phi$, the exhaustive search approach is only effective in certain special scenarios where either the number of SD pairs in a time slot is small or the number of candidate routes for each SD pair is small. Although these special scenarios have practical significance (e.g., when the QDN have pre-computed only one or a small number of candidate routes for each SD pair or when the EC request rate is low), in the remainder of this section, we develop a low-complexity approach based on Gibbs sampling (GS) \cite{geman1984stochastic} that can handle the general scenario.

Our algorithm follows an iterative process as described in Algorithm~\ref{alg:route_selection} and operates as follows. Initially, each SD pair $\phi \in \Phi$ randomly selects a route $r^0(\phi) \in \mathcal{R}(\phi)$ from the candidate set, and the qubit allocation $\mathcal{N}^(r^0)$ is computed. In each iteration $k$, one SD pair, chosen randomly (say $\phi$), virtually modifies its current route selection to $\tilde{r}(\phi)$, while the routes for the remaining SD pairs remain unchanged. Denoting $\tilde{\boldsymbol{r}}^k = (\tilde{r}(\phi), \{r^{k-1}(\phi')\}_{\phi' \neq \phi})$, the qubits are then allocated according to the algorithm from the previous subsection, resulting in the objective function value $f(\tilde{\boldsymbol{r}}^k, \mathcal{N}^(\tilde{\boldsymbol{r}}^k))$. The algorithm proceeds by computing the difference in objective function values: $f(\tilde{\boldsymbol{r}}^k, \mathcal{N}^*(\tilde{\boldsymbol{r}}^k)) - f(\boldsymbol{r}^{k-1}, \mathcal{N}^*(\tilde{\boldsymbol{r}}^{k-1}))$, and then updating the route selection decision based on this difference. Specifically, with probability $\eta$, the route selection remains the same as in the previous iteration, i.e., $\boldsymbol{r}^k = \boldsymbol{r}^{k-1}$, and with probability $1-\eta$, the route selection changes, i.e., $\boldsymbol{r}^k = \tilde{\boldsymbol{r}}^k$. The probability $\eta$ is calculated as:
\begin{align}
\label{eta}
\eta = \left((1 + \exp{\frac{f(\tilde{\boldsymbol{r}}^k, \mathcal{N}^*(\tilde{\boldsymbol{r}}^k)) - f(\boldsymbol{r}^{k-1}, \mathcal{N}^*(\tilde{\boldsymbol{r}}^{k-1}))}{\gamma}}\right)^{-1},
\end{align}
where $\gamma > 0$ is a parameter controlling the degree of exploration versus exploitation (i.e., the level of randomness). Consequently, changing the route selection is more likely to occur if the new route selection $\tilde{\boldsymbol{r}}^k$ results in a higher objective value.

\textbf{Remark}: (1) In combinatorial optimization, it is widely recognized that purely greedily selecting better decisions can often lead to local optima. To avoid getting stuck in local optima, our algorithm incorporates a probabilistic exploration mechanism, even when it may result in worse performance than the current decision. However, such an algorithm is known to possess an asymptotic convergence rate. Specifically, as $\gamma \to 0$, the algorithm converges to the global optimal solution with a probability of 1. (2) The convergence rate of our algorithm can be further improved by allowing multiple SD pairs to simultaneously evolve their route selection in each iteration, as long as they do not share common edges in the candidate route set. This simultaneous evolution is particularly beneficial when the SD pairs are spatially far apart, as it enables them to explore different parts of the solution space concurrently. By diversifying the exploration, our algorithm becomes more efficient in converging towards the global optimum.

\subsection{Performance Analysis}
In this section, we present a theoretical performance analysis for the proposed quantum routing and qubit allocation algorithm. As a reminder, our algorithm for the per-slot problem produces a $\Delta$-optimal solution at each time slot $t$. The level of $\Delta$-optimality achieved is exact when route selection employs exhaustive search (which is effective with a small $F$ and a small $R$) and asymptotic when route selection employs the low-complexity algorithm from the previous subsection. Furthermore, we introduce the following assumption, which is essential for the feasibility of the system:
\begin{assumption}
The cost budget $C$ satisfies $C \geq FLT$.
\end{assumption}
This assumption ensures that there is a sufficient cost budget to establish at least one quantum route for each source-destination (SD) pair in each time slot. In the worst case scenario, each link in the route utilizes at least one quantum channel. Meeting this budget requirement is crucial for the minimum operation of the system.

\begin{theorem}
    Under Assumption 1, solving the per-slot problem \textbf{P2} with $\Delta$-optimality in each time slot $t$ ensures the following bound on the constraint violation:
    \begin{align}
    \frac{1}{T}\sum_{t=0}^{T-1} c^t - \frac{C}{T} \leq \sqrt{\frac{(q^0)^2}{T^2} + \frac{2(\Delta + B - VFL\log(p^\text{min}))}{T}} - \frac{q^0}{T}.
    \end{align}
\end{theorem}
\begin{proof}
The Lyapunov drift of the virtual queue length is
\begin{align}
    \delta(t) =& \frac{1}{2}[(q^{t+1})^2 - (q^t)^2]  \leq  \frac{1}{2} [(q^t + c^t - C/T)^2 - (q^t)^2]\nonumber\\
    = &q^t (c^t - C/T) + \frac{1}{2}(c^t - C/T)^2
    = q^t (c^t - C/T) + B, 
    \label{delta_bound}
\end{align}
where $B$ is finite constant and it exists because the qubit/channel capacity is upper bounded. 
\begin{align}
&V\cdot\sum_{\phi\in \Phi^t}\log P(r^t(\phi), \mathcal{N}^t(r^t(\phi))) - \delta(t)\nonumber\\
\geq& V\cdot\sum_{\phi\in \Phi^t}\log P(r^t(\phi), \mathcal{N}^t(r^t(\phi))) - q^t\cdot \sum_{\phi \in \Phi^t}\sum_{e \in r^t(\phi)} n^t_e(r^t(\phi))) \nonumber\\
&+ q^t C/T - B \nonumber\\
\geq & \sum_{\phi\in \Phi^t}\sum_{e \in \tilde{r}^t(\phi)}(V\log P_e(1) - q^t\cdot 1) - \Delta + q^t C/T - B \nonumber\\
\geq &VFL\log(p^\text{min}) + q^t(CT - VFL) - \Delta - B \nonumber\\
\geq &VFL\log(p^\text{min}) - \Delta - B.
\end{align}
The first inequality obtained is by plugging the bound \eqref{delta_bound} on $\delta(t)$. The second inequality holds due to the $\Delta$-optimality of our per-slot solution. Specifically, we consider a particular solution that chooses route $\tilde{r}(\phi)$ for SD pair $\phi \in \Phi$, and uses a single quantum channel on each edge. Clearly, the route selection and qubit allocation produced by our algorithm performs no worse than this solution by a constant $\Delta$. The third inequality is due to the monotonicity of the entanglement success rate and $V\log P_e(1) - q^t < 0$. The last inequality holds due to Assumption 1. 

Denote $D \triangleq \Delta + B - VFL\log(p^\text{min}) > 0$. Since $\sum_{\phi\in \Phi^t}\log P(r^t(\phi), \mathcal{N}^t(r^t(\phi))) < 0$ always holds, we have $\delta(t) \leq D$. 
Because $\frac{1}{2}[(q^T)^2 - (q^0)^2] = \sum_{t=0}^{T-1}\delta(t) \leq D T$, we then have $q^T \leq \sqrt{(q^0)^2 + 2DT}$. According to the virtual queue dynamics, $q^{t+1} \geq q^t + c^t - C/T$. This leads to
\begin{align}
    &\frac{1}{T}\sum_{t=0}^{T-1} c^t - \frac{C}{T} \leq \frac{1}{T}\sum_{t=0}^{T-1}(q^{t+1} - q^t) = \frac{q^T - q^0}{T}\nonumber\\
    \leq & \sqrt{\frac{(q^0)^2}{T^2} + \frac{2D}{T}} - \frac{q^0}{T}.
\end{align}
\end{proof}

\textbf{Remarks:} Theorem 1 provides valuable insights into the budget constraint's behavior with respect to system parameters, the initial queue length $q^0$ and the parameter $V$. As $T$ approaches infinity, the budget constraint is asymptotically satisfied. For a finite $T$, the right-hand-side (RHS) of the constraint decreases as the initial queue length $q^0$ increases. In fact, in the limit as $q^0$ approaches infinity, the RHS tends to 0. This intuitive behavior occurs because a large initial queue length results in our algorithm applying a significant penalty to any violation of the budget constraint, leading to a highly conservative qubit allocation strategy. Moreover, Theorem 1 highlights the impact of the parameter $V$ on the budget constraint. With a larger value for $V$, the potential violation becomes greater (noting that $\log(p^\text{min}) < 0$). This observation aligns with intuition, as a higher value of $V$ indicates that the algorithm places more emphasis on improving the entanglement success rate. As a consequence, the algorithm becomes less cautious in managing the qubit allocation to avoid any budget constraint violation.

\begin{theorem}
    Under Assumption 1, solving the per-slot problem \text{P2} with $\Delta$-optimality in each time slot $t$ ensures the following bounds on the objective of \textbf{P1}
    \begin{align}
        \frac{1}{T}\sum_{t=0}^{T-1} \mathbb{E}[u(\boldsymbol{r}^t, \mathcal{N}^t)] \geq \text{OPT} - \frac{\Delta + B}{V} - \frac{(q^0)^2}{2VT},
    \end{align}
    where $\text{OPT}$ denotes the expected optimal value of the time-averaged objective given by a possibly randomized offline algorithm that has complete statistics of all $T$ time slots. 
\end{theorem}
\begin{proof}
The proof mostly follows Theorem 4.8 in \cite{neely2022stochastic} by incorporating the $q^0$ term. 


Denote $\hat{\boldsymbol{r}}^t$ and $\hat{\mathcal{N}}^t$ as the joint route selection and qubit selection outcome produced by the offline optimal algorithm that achieves $\text{OPT}$. Consider the objective function minus drift
\begin{align}
    &V\cdot \sum_{t=0}^{T-1}u(\boldsymbol{r}^t, \mathcal{N}^t) - \sum_{t=0}^{T-1}\delta(t)\nonumber\\
    \geq & V \cdot \sum_{t=0}^{T-1}u(\boldsymbol{r}^t, \mathcal{N}^t) - \sum_{t=0}^{T-1}q^t(c^t - C/T) - BT\nonumber\\
    = &\sum_{t=0}^{T-1}\left(V\cdot u(\boldsymbol{r}^t, \mathcal{N}^t) - q^t c^t\right) + \sum_{t=0}^{T-1} q^t C/T - BT\nonumber\\
    \geq &\sum_{t=0}^{T-1}\left(V\cdot u(\hat{\boldsymbol{r}}^t, \hat{\mathcal{N}}^t) - q^t \hat{c}^t - \Delta \right) + \sum_{t=0}^{T-1} q^t C/T - BT\nonumber\\
    = & V\cdot\sum_{t=0}^{T-1} u(\hat{\boldsymbol{r}}^t, \hat{\boldsymbol{N}}^t)  + \sum_{t=0}^{T-1} q^t (C/T - \hat{c}^t) - (\Delta+B)T.
\end{align}
The first inequality is due to the bound \eqref{delta_bound} on $\delta(t)$. The second inequality is because $\boldsymbol{r}^t$ and $\mathcal{N}^t$ produces a $\Delta$-optimal solution given $q^t$ in time slot $t$. Here, $\hat{c}^t$ is the cost incurred in time slot $t$ by the offline optimal solution.  

Now, taking expectation over the system randomness on both sides, we have
\begin{align}
&V\cdot \sum_{t=0}^{T-1}\mathbb{E}[u(\boldsymbol{r}^t, \mathcal{N}^t)] - \sum_{t=0}^{T-1}\delta(t) \nonumber\\
\geq &V\cdot \sum_{t=0}^{T-1} \mathbb{E}[u(\hat{\boldsymbol{r}}^t, \hat{\boldsymbol{N}}^t)] + \sum_{t=0}^{T-1} q^t (C/T - \mathbb{E}[\hat{c}^t]) - (\Delta + B)T \nonumber\\
= &V\cdot \text{OPT} - (\Delta + B)T,
\end{align}
where the last equality utilizes $\mathbb{E}[\hat{c}^t] = C/T$ for the optimal offline algorithm, which is independent of $q^t$. Finally, noticing $\sum_{t=0}^{T-1}\delta(t) \geq -\frac{1}{2}(q^0)^2$ and moving it to the right hand side yields
\begin{align}
&V\cdot \sum_{t=0}^{T-1}\mathbb{E}[u(\boldsymbol{r}^t, \mathcal{N}^t)] \geq V\cdot \text{OPT} - (\Delta + B)T -\frac{1}{2}(q^0)^2.
\end{align}
Dividing both sides by $VT$ yields the final result. 
\end{proof}

\textbf{Remark}: Theorem 2 provides important insights into the performance of our algorithm, showing that the expected objective value achieved by running the algorithm is approximately optimal, with a constant gap that depends on various factors. Firstly, the optimality gap is influenced by the system parameters, the initial queue length $q^0$, and the parameter $V$. Specifically, with a larger $q^0$, the optimality gap becomes larger. This observation is intuitive as a larger initial queue length prompts the algorithm to act more cautiously in avoiding any budget constraint violation, which could potentially lead to a sacrifice in entanglement performance. However, as the time horizon $T$ approaches infinity, the impact of the initial queue length diminishes, suggesting that over a sufficiently long duration, the algorithm's performance becomes less sensitive to the initial queue length. Secondly, the optimality gap is affected by the parameter $V$. A larger value for $V$ results in a reduced optimality gap. The reason behind this is that the algorithm places a higher emphasis on improving the entanglement performance with an increased $V$, thereby prioritizing the entanglement success rate over budget constraint satisfaction.

\section{Simulation Results}\label{s5}

In this section, we conduct an evaluation of our proposed algorithm OSCAR while also comparing its performance against several baseline methods. Additionally, we delve into an ablation study to assess the impact of various control parameters.

\subsection{Simulation Setup}
\subsubsection{Network Topology}
To create a random QDN topology for benchmarking purposes, we adopt the following procedure. Initially, we select a predefined number of nodes within a $100\times100$ unit square area,. Subsequently, we utilize the Waxman graph \cite{waxman1988routing} to generate the specific topology. This model establishes edges between nodes $u$ and $v$ with a probability of $\beta \exp \left(-\frac{d(u,v)}{\alpha d_{\max}}\right)$, where $d(u,v)$ denotes the Euclidean distance between nodes $u$ and $v$, and $d_{\max}$ represents the longest distance among any pair of nodes in the network. The two control parameters $\alpha$ and $\beta$ govern the characteristics of the generated topology. It is important to note that this generation approach has been previously employed in various works focusing on quantum networks \cite{zhao2021redundant,li2022fidelity}.

\subsubsection{Default Parameters}
In the default configuration, the QDN consists of $20$ nodes, each with a random qubit capacity following a discrete uniform distribution $\mathcal{U}[10,16]$. The generation of edges between nodes is performed using $\alpha=0.5$ and $\beta=0.5$, resulting in an average node degree of approximately 4. The channel capacity of each generated edge is also chosen randomly from $\mathcal{U}[5,8]$. The probability of successfully establishing a quantum link is $2\times 10^{-4}$ per attempt, and during each time slot, $4000$ attempts can be made. The user operates under a total qubit cost budget of $C=5000$ over a period of $T=200$ time slots. Additionally, the number of source-destination (SD) pairs varies randomly in each time slot, following $\mathcal{U}[1,5]$. For the Lyapunov control parameter, we set $V = 2500$, and the initial value is $q^0 = 10$ by default. Furthermore, the GS parameter $\gamma$ is set to $500$. To obtain reliable results, we conduct 5 trial simulations and present the averaged outcomes.

\subsubsection{Baseline Schemes}
We compare our proposed algorithm OSCAR with two  baselines:

\textbf{Myopic-Fixed (MF)}: In this approach, we adopt a uniform allocation strategy, evenly distributing the total budget among each time slot. Consequently, the user's budget available for utilization in each time slot is given by $C/T$. Subsequently, we address the maximization problem on a per-slot basis, ensuring adherence to the budget constraint for each slot. 

\textbf{Myopic-Adaptive (MA)}: One drawback of MF is the potential for budget waste, as the allocation to each time slot may not be fully utilized. To tackle this concern, in MA, any remaining budget is uniformly distributed among the remaining time slots. Specifically, the budget available for the user in time slot $t$ becomes $(C - \tilde{C})/(T-t)$, where $\tilde{C}$ represents the total number of qubits consumed up to that point. The per-slot entanglement routing problem in MA is resolved in a manner similar to MF.

\subsection{Performance Comparison}
In this section, we compare the performance of OSCAR with the two baseline approaches under the default configuration.  
\subsubsection{Time-evolving Performance}

\begin{figure*}[th]
	\centering
    \begin{minipage}[b]{0.72\textwidth}
	\hspace{-4mm}
	\subfigure[Average Utility]{
		\begin{minipage}[b]{0.33\textwidth}
			\includegraphics[width=1\textwidth]{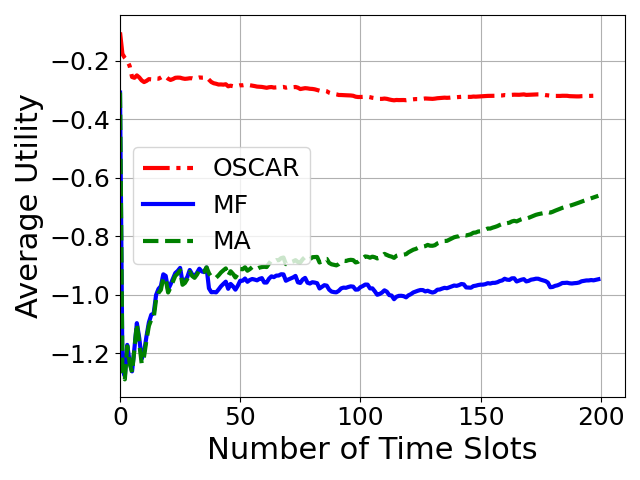}
		\end{minipage}
		\label{fig:main_utility}
	}
	\hspace{-4mm}
    	\subfigure[Average Success Rate]{
    		\begin{minipage}[b]{0.33\textwidth}
   		 	\includegraphics[width=1\textwidth]{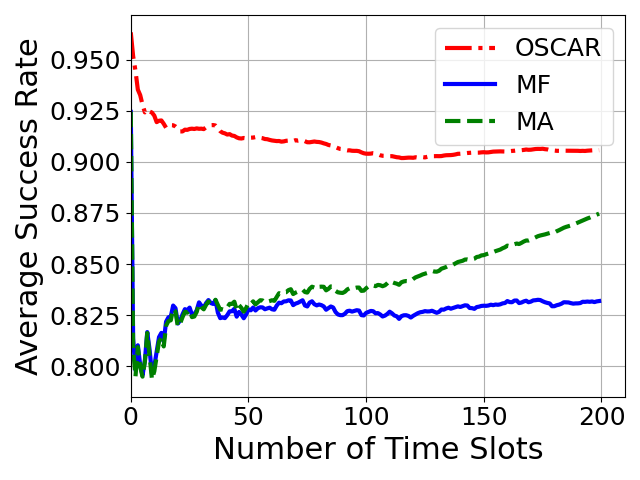}ev
    		\end{minipage}
		\label{fig:main_prob}
    	}
    \hspace{-4mm}
    	\subfigure[Average Qubits Usage]{
		    \begin{minipage}[b]{0.33\textwidth}
   	 	    \includegraphics[width=1\textwidth]{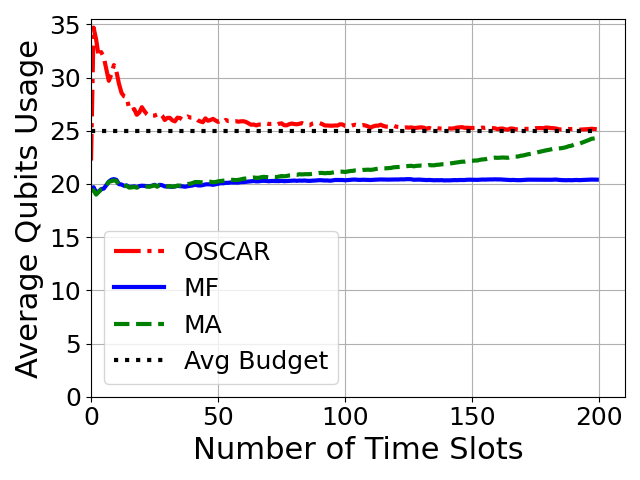}
		    \end{minipage}
	    \label{fig:main_budget}
	    }
	\caption{Time-evolving Performance of Different Methods.}
	\label{fig:main}
	 \vspace{-10 pt}
	\end{minipage}
    \vspace{-10 pt}
    \begin{minipage}[b]{0.25\textwidth}
   	 	\includegraphics[width=1\textwidth]{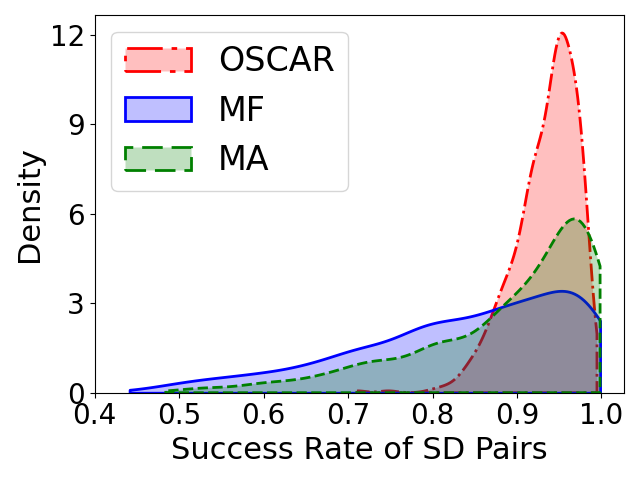}
   	\caption{Success Rate Distribution.}
   	\label{fig:main_density}
    \end{minipage}
\end{figure*}

Figure~\ref{fig:main} illustrates the time-evolving performance of OSCAR and two baseline methods in a specific experiment run, showcasing the average utility, average EC success rate, and average qubit usage. Notably, OSCAR outperforms MA and MF with a significantly higher utility and EC success rate (i.e., 0.9) while effectively adhering to the qubit budget constraint by the end of $T = 200$ time slots. Conversely, the myopic fixed budget allocation approach used by MF leads to under-utilization of the qubit budget, resulting in a much lower utility and EC success rate (i.e., 0.83). Although MA eventually achieves a similar qubit usage to OSCAR by the end of $T = 200$ time slots, its final utility and EC success rate (i.e., 0.875) remain noticeably lower than those of OSCAR. Of particular concern is the observation that the conservative qubit allocation employed by MA in the early time slots significantly lowers the average utility and EC success rate, compared to the late time slots, implying an unfair distribution of qubit resources among SD pairs over time. To further validate the fairness of our method, Figure~\ref{fig:main_density} presents the success rate distribution of different methods, confirming that OSCAR ensures a much more equitable allocation among the SD pairs.

\begin{figure}[t]
	\centering
    	\subfigure[Average Success Rate]{
    		\begin{minipage}[b]{0.45\linewidth}
   		 	\includegraphics[width=1\linewidth]{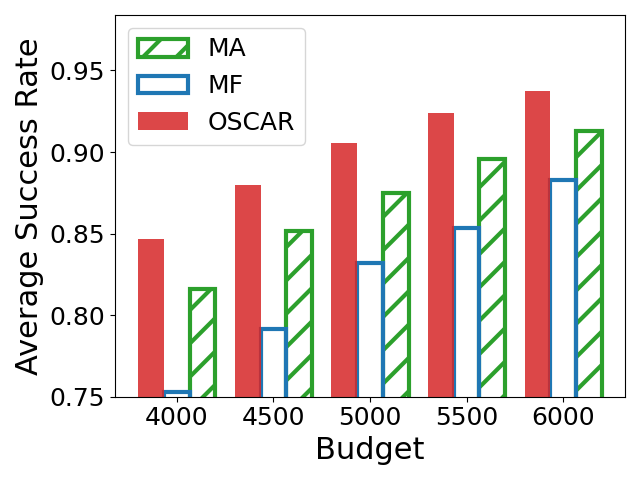}
    \vspace{-15 pt}       
    		\end{minipage}
		\label{fig:c_prob}

    	}
    	\subfigure[Average Qubits Usage]{
		    \begin{minipage}[b]{0.45\linewidth}
   	 	    \includegraphics[width=1\linewidth]{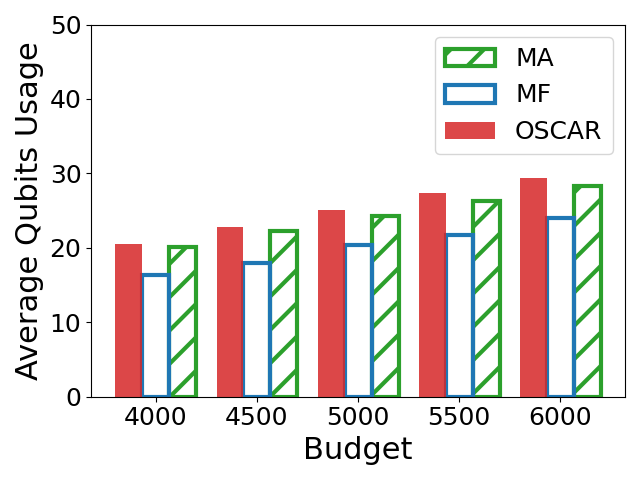}
      \vspace{-15 pt}          
		    \end{minipage}
	    \label{fig:c_budget}
	    }
	\caption{Impact of Budget.}
	\label{fig:c}
 \vspace{-10 pt}
\end{figure}

\subsubsection{Impact of Budget}
In Figure~\ref{fig:c}, we explore how the qubit budget $C$ influences the EC performance of various methods. As expected, all approaches exhibit improved EC success rates as the budget is increased, demonstrating the positive impact of allocating more qubits for establishing entanglement links. Notably, OSCAR consistently outperforms both MF and MA across different budget levels. However, it is worth noting that the performance gap between OSCAR and the baseline approaches diminishes as the budget becomes larger. This result is intuitive since ample resources allow for easier establishment of entanglement links with higher success rates through the allocation of more qubits. Consequently, this emphasizes the importance of judiciously distributing qubits, particularly when the budget is limited.

 \begin{figure}[t]
 \centering
    	\subfigure[Average Success Rate]{
    		\begin{minipage}[b]{0.45\linewidth}
   		 	\includegraphics[width=1\linewidth]{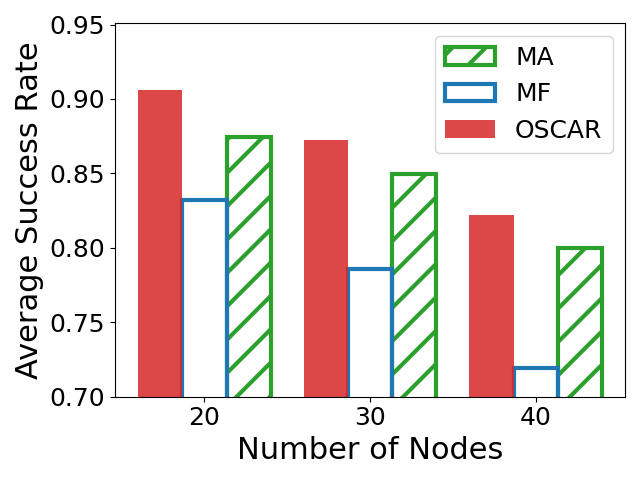}
      \vspace{-10 pt}       
    		\end{minipage}
		\label{fig:n_prob}
    	}
    	\subfigure[Average Qubits Usage]{
		    \begin{minipage}[b]{0.45\linewidth}
   	 	    \includegraphics[width=1\linewidth]{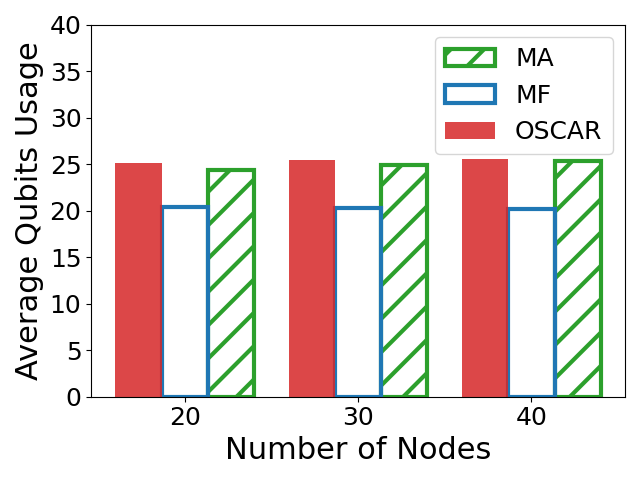}
      \vspace{-10 pt}          
		    \end{minipage}
	    \label{fig:n_budget}
	    }
	\caption{Impact of Network Size.}
	\label{fig:n}
 \vspace{-5 pt}
\end{figure}

\subsubsection{Impact of Network Size}
Figure~\ref{fig:n} demonstrates the impact of network size on the EC performance of different methods under the same qubit budget. We adjust the Waxman graph parameter to ensure an average node degree of approximately 4 across all network sizes.
As anticipated, all methods experience reduced EC success rates as the network size increases. This decline can be attributed to the longer routes required to connect the source and destination nodes in larger networks. However, even in the face of this challenge, OSCAR consistently outperforms both MF and MA across various network scales. This emphasizes the potential benefits of employing OSCAR in real-world scenarios where network scales may vary.

\subsection{Impact of Algorithm Parameters} \label{sec:ablation}
\begin{figure}[t]
	\centering
    \begin{minipage}[b]{0.24\textwidth}
        \includegraphics[width=1\textwidth]{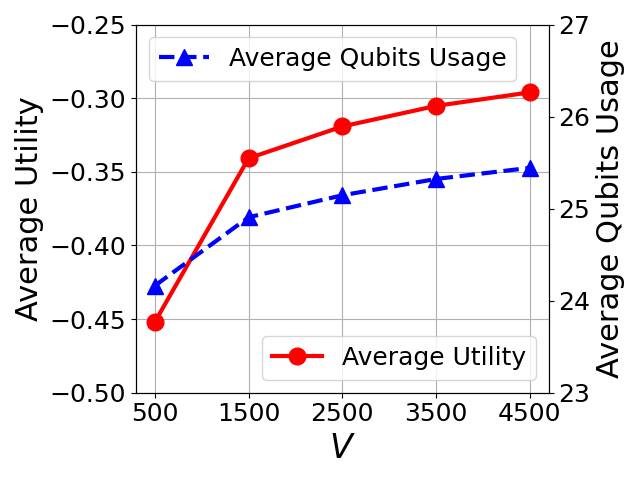}
   	\caption{Impact of $V$.}
   	\label{fig:V}
	\end{minipage}
    \vspace{-10 pt}
    \begin{minipage}[b]{0.24\textwidth}
        \includegraphics[width=1\textwidth]{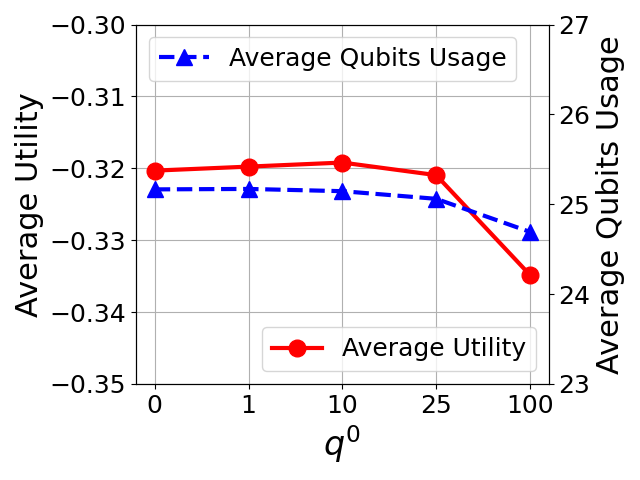}
   	\caption{Impact of $q^0$.}
   	\label{fig:q0}
	\end{minipage}
\end{figure} 

We now investigate the impact of algorithm parameters of OSCAR on its performance. 

\subsubsection{Impact of control parameter $V$}
The control parameter $V$ plays a crucial role in striking a balance between maximizing entanglement performance and adhering to the budget constraint. When $V$ is larger, greater emphasis is placed on maximizing utility, while a smaller $V$ prioritizes staying within the budget constraint. As demonstrated in Figure~\ref{fig:V}, tuning $V$ effectively achieves this trade-off. With an increased value of $V$, OSCAR attains higher utility, which highlights its capability to optimize entanglement performance. However, it is observed that a larger $V$ also leads to a greater violation of the qubit budget constraint, indicating that a higher focus on performance optimization may come at the expense of exceeding the allocated qubit resources. This is aligned with our theoretical results. Hence, the choice of $V$ is crucial and should be tailored to the specific requirements of the application. 

\subsubsection{Impact of initial virtual queue size $q^0$}
Lastly, Figure~\ref{fig:q0} illustrates the impact of the initial virtual queue value $q^0$ on both entanglement utility and qubit usage. Our simulation results align with the theoretical analysis, demonstrating that a larger $q^0$ leads to reduced qubit usage by limiting resource allocation during the initial time slots. However, setting $q^0$ too large can have a negative effect on entanglement performance. Our findings indicate that a relatively small $q^0$, in contrast to the common practice of setting $q^0 = 0$, proves to be effective in reducing qubit resource usage while maintaining a nearly stable entanglement utility.

\section{Conclusion}\label{s6}
In this paper, we investigated a user-centric entanglement routing problem in QDNs, a critical facilitator for distributed quantum computing. Our primary focus was on the cost aspect of utilizing the QDN to establish quantum entanglement links, emphasizing the significance of considering the user's long-term budget constraint while making entanglement routing decisions.
To address this challenge, we devised a novel adaptive entanglement routing algorithm, which enables the efficient discovery of quantum routes between source and destination quantum nodes, while also allocating appropriate qubit resources along these routes in an online fashion. Importantly, our algorithm achieves this without relying on future entanglement connection request statistics or network resource dynamics. Theoretical analysis and extensive simulation studies were conducted. The results demonstrated the superiority of our approach compared to baselines that adopt a myopic entanglement routing strategy for current requests only. Our approach offers practical benefits in terms of cost optimization and resource allocation, paving the way for more efficient and scalable quantum communication and computation across distributed quantum systems.

\section*{Acknowledgment}
To improve readability and quality of language, this paper has been grammatically revised using ChatGPT.

\bibliographystyle{IEEEtran}
\bibliography{main}

\begin{thebibliography}{10}
\providecommand{\url}[1]{#1}
\csname url@samestyle\endcsname
\providecommand{\newblock}{\relax}
\providecommand{\bibinfo}[2]{#2}
\providecommand{\BIBentrySTDinterwordspacing}{\spaceskip=0pt\relax}
\providecommand{\BIBentryALTinterwordstretchfactor}{4}
\providecommand{\BIBentryALTinterwordspacing}{\spaceskip=\fontdimen2\font plus
\BIBentryALTinterwordstretchfactor\fontdimen3\font minus
  \fontdimen4\font\relax}
\providecommand{\BIBforeignlanguage}[2]{{%
\expandafter\ifx\csname l@#1\endcsname\relax
\typeout{** WARNING: IEEEtran.bst: No hyphenation pattern has been}%
\typeout{** loaded for the language `#1'. Using the pattern for}%
\typeout{** the default language instead.}%
\else
\language=\csname l@#1\endcsname
\fi
#2}}
\providecommand{\BIBdecl}{\relax}
\BIBdecl

\bibitem{steane1998quantum}
A.~Steane, ``Quantum computing,'' \emph{Reports on Progress in Physics},
  vol.~61, no.~2, p. 117, 1998.

\bibitem{cacciapuoti2019quantum}
A.~S. Cacciapuoti, M.~Caleffi, F.~Tafuri, F.~S. Cataliotti, S.~Gherardini, and
  G.~Bianchi, ``Quantum internet: Networking challenges in distributed quantum
  computing,'' \emph{IEEE Network}, vol.~34, no.~1, pp. 137--143, 2019.

\bibitem{cirac1999distributed}
J.~I. Cirac, A.~Ekert, S.~F. Huelga, and C.~Macchiavello, ``Distributed quantum
  computation over noisy channels,'' \emph{Physical Review A}, vol.~59, no.~6,
  p. 4249, 1999.

\bibitem{buhrman2003distributed}
H.~Buhrman and H.~R{\"o}hrig, ``Distributed quantum computing,'' in
  \emph{International Symposium on Mathematical Foundations of Computer
  Science}.\hskip 1em plus 0.5em minus 0.4em\relax Springer, 2003, pp. 1--20.

\bibitem{mao2023qubit}
Y.~Mao, Y.~Liu, and Y.~Yang, ``Qubit allocation for distributed quantum
  computing,'' in \emph{IEEE INFOCOM 2023-IEEE Conference on Computer
  Communications}.\hskip 1em plus 0.5em minus 0.4em\relax IEEE, 2023, pp.
  1--10.

\bibitem{yang2023asynchronous}
L.~Yang, Y.~Zhao, L.~Huang, and C.~Qiao, ``Asynchronous entanglement
  provisioning and routing for distributed quantum computing,'' in \emph{IEEE
  INFOCOM 2023-IEEE Conference on Computer Communications}.\hskip 1em plus
  0.5em minus 0.4em\relax IEEE, 2023.

\bibitem{mehic2020quantum}
M.~Mehic, M.~Niemiec, S.~Rass, J.~Ma, M.~Peev, A.~Aguado, V.~Martin,
  S.~Schauer, A.~Poppe, C.~Pacher \emph{et~al.}, ``Quantum key distribution: a
  networking perspective,'' \emph{ACM Computing Surveys (CSUR)}, vol.~53,
  no.~5, pp. 1--41, 2020.

\bibitem{bouwmeester1997experimental}
D.~Bouwmeester, J.-W. Pan, K.~Mattle, M.~Eibl, H.~Weinfurter, and A.~Zeilinger,
  ``Experimental quantum teleportation,'' \emph{Nature}, vol. 390, no. 6660,
  pp. 575--579, 1997.

\bibitem{wootters1982single}
W.~K. Wootters and W.~H. Zurek, ``A single quantum cannot be cloned,''
  \emph{Nature}, vol. 299, no. 5886, pp. 802--803, 1982.

\bibitem{vedral2014quantum}
V.~Vedral, ``Quantum entanglement,'' \emph{Nature Physics}, vol.~10, no.~4, pp.
  256--258, 2014.

\bibitem{stephenson2020high}
L.~Stephenson, D.~Nadlinger, B.~Nichol, S.~An, P.~Drmota, T.~Ballance,
  K.~Thirumalai, J.~Goodwin, D.~Lucas, and C.~Ballance, ``High-rate,
  high-fidelity entanglement of qubits across an elementary quantum network,''
  \emph{Physical review letters}, vol. 124, no.~11, p. 110501, 2020.

\bibitem{jennewein2001experimental}
T.~Jennewein, G.~Weihs, J.-W. Pan, and A.~Zeilinger, ``Experimental nonlocality
  proof of quantum teleportation and entanglement swapping,'' \emph{Physical
  review letters}, vol.~88, no.~1, p. 017903, 2001.

\bibitem{hilaire2021error}
P.~Hilaire, E.~Barnes, S.~E. Economou, and F.~Grosshans, ``Error-correcting
  entanglement swapping using a practical logical photon encoding,''
  \emph{Physical Review A}, vol. 104, no.~5, p. 052623, 2021.

\bibitem{schoute2016shortcuts}
E.~Schoute, L.~Mancinska, T.~Islam, I.~Kerenidis, and S.~Wehner, ``Shortcuts to
  quantum network routing,'' \emph{arXiv preprint arXiv:1610.05238}, 2016.

\bibitem{pant2019routing}
M.~Pant, H.~Krovi, D.~Towsley, L.~Tassiulas, L.~Jiang, P.~Basu, D.~Englund, and
  S.~Guha, ``Routing entanglement in the quantum internet,'' \emph{npj Quantum
  Information}, vol.~5, no.~1, p.~25, 2019.

\bibitem{chakraborty2019distributed}
K.~Chakraborty, F.~Rozpedek, A.~Dahlberg, and S.~Wehner, ``Distributed routing
  in a quantum internet,'' \emph{arXiv preprint arXiv:1907.11630}, 2019.

\bibitem{vardoyan2019stochastic}
G.~Vardoyan, S.~Guha, P.~Nain, and D.~Towsley, ``On the stochastic analysis of
  a quantum entanglement switch,'' \emph{ACM SIGMETRICS Performance Evaluation
  Review}, vol.~47, no.~2, pp. 27--29, 2019.

\bibitem{shi2020concurrent}
S.~Shi and C.~Qian, ``Concurrent entanglement routing for quantum networks:
  Model and designs,'' in \emph{Proceedings of the Annual conference of the ACM
  Special Interest Group on Data Communication on the applications,
  technologies, architectures, and protocols for computer communication}, 2020,
  pp. 62--75.

\bibitem{zeng2022multi}
Y.~Zeng, J.~Zhang, J.~Liu, Z.~Liu, and Y.~Yang, ``Multi-entanglement routing
  design over quantum networks,'' in \emph{IEEE INFOCOM 2022-IEEE Conference on
  Computer Communications}.\hskip 1em plus 0.5em minus 0.4em\relax IEEE, 2022,
  pp. 510--519.

\bibitem{zhao2021redundant}
Y.~Zhao and C.~Qiao, ``Redundant entanglement provisioning and selection for
  throughput maximization in quantum networks,'' in \emph{IEEE INFOCOM
  2021-IEEE Conference on Computer Communications}.\hskip 1em plus 0.5em minus
  0.4em\relax IEEE, 2021, pp. 1--10.

\bibitem{zhao2022segmented}
G.~Zhao, J.~Wang, Y.~Zhao, H.~Xu, and C.~Qiao, ``Segmented entanglement
  establishment for throughput maximization in quantum networks,'' in
  \emph{2022 IEEE 42nd International Conference on Distributed Computing
  Systems (ICDCS)}.\hskip 1em plus 0.5em minus 0.4em\relax IEEE, 2022, pp.
  45--55.

\bibitem{zhao2022e2e}
Y.~Zhao, G.~Zhao, and C.~Qiao, ``E2e fidelity aware routing and purification
  for throughput maximization in quantum networks,'' in \emph{IEEE INFOCOM
  2022-IEEE Conference on Computer Communications}.\hskip 1em plus 0.5em minus
  0.4em\relax IEEE, 2022, pp. 480--489.

\bibitem{caleffi2017optimal}
M.~Caleffi, ``Optimal routing for quantum networks,'' \emph{Ieee Access},
  vol.~5, pp. 22\,299--22\,312, 2017.

\bibitem{li2022fidelity}
J.~Li, M.~Wang, K.~Xue, R.~Li, N.~Yu, Q.~Sun, and J.~Lu, ``Fidelity-guaranteed
  entanglement routing in quantum networks,'' \emph{IEEE Transactions on
  Communications}, vol.~70, no.~10, pp. 6748--6763, 2022.

\bibitem{yang2022online}
L.~Yang, Y.~Zhao, H.~Xu, and C.~Qiao, ``Online entanglement routing in quantum
  networks,'' in \emph{2022 IEEE/ACM 30th International Symposium on Quality of
  Service (IWQoS)}.\hskip 1em plus 0.5em minus 0.4em\relax IEEE, 2022, pp.
  1--10.

\bibitem{farahbakhsh2022opportunistic}
A.~Farahbakhsh and C.~Feng, ``Opportunistic routing in quantum networks,'' in
  \emph{IEEE INFOCOM 2022-IEEE Conference on Computer Communications}.\hskip
  1em plus 0.5em minus 0.4em\relax IEEE, 2022, pp. 490--499.

\bibitem{van2014quantum}
R.~Van~Meter, \emph{Quantum networking}.\hskip 1em plus 0.5em minus 0.4em\relax
  John Wiley \& Sons, 2014.

\bibitem{mattle1996dense}
K.~Mattle, H.~Weinfurter, P.~G. Kwiat, and A.~Zeilinger, ``Dense coding in
  experimental quantum communication,'' \emph{Physical review letters},
  vol.~76, no.~25, p. 4656, 1996.

\bibitem{nielsen2001quantum}
M.~A. Nielsen and I.~L. Chuang, ``Quantum computation and quantum
  information,'' \emph{Phys. Today}, vol.~54, no.~2, p.~60, 2001.

\bibitem{pirandola2015advances}
S.~Pirandola, J.~Eisert, C.~Weedbrook, A.~Furusawa, and S.~L. Braunstein,
  ``Advances in quantum teleportation,'' \emph{Nature photonics}, vol.~9,
  no.~10, pp. 641--652, 2015.

\bibitem{dahlberg2019link}
A.~Dahlberg, M.~Skrzypczyk, T.~Coopmans, L.~Wubben, F.~Rozpedek, M.~Pompili,
  A.~Stolk, P.~Pawelczak, R.~Knegjens, J.~de~Oliveira~Filho \emph{et~al.}, ``A
  link layer protocol for quantum networks,'' in \emph{Proceedings of the ACM
  special interest group on data communication}, 2019, pp. 159--173.

\bibitem{dijkstra2022note}
E.~W. Dijkstra, ``A note on two problems in connexion with graphs,'' in
  \emph{Edsger Wybe Dijkstra: His Life, Work, and Legacy}, 2022, pp. 287--290.

\bibitem{kelly1997charging}
F.~Kelly, ``Charging and rate control for elastic traffic,'' \emph{European
  transactions on Telecommunications}, vol.~8, no.~1, pp. 33--37, 1997.

\bibitem{boyd2004convex}
S.~P. Boyd and L.~Vandenberghe, \emph{Convex optimization}.\hskip 1em plus
  0.5em minus 0.4em\relax Cambridge university press, 2004.

\bibitem{geman1984stochastic}
S.~Geman and D.~Geman, ``Stochastic relaxation, gibbs distributions, and the
  bayesian restoration of images,'' \emph{IEEE Transactions on pattern analysis
  and machine intelligence}, no.~6, pp. 721--741, 1984.

\bibitem{neely2022stochastic}
M.~Neely, \emph{Stochastic network optimization with application to
  communication and queueing systems}.\hskip 1em plus 0.5em minus 0.4em\relax
  Springer Nature, 2022.

\bibitem{waxman1988routing}
B.~M. Waxman, ``Routing of multipoint connections,'' \emph{IEEE journal on
  selected areas in communications}, vol.~6, no.~9, pp. 1617--1622, 1988.

\end{thebibliography}

\end{document}